\documentclass[11pt,oneside]{article}
\usepackage[paper=letterpaper,margin=1in]{geometry}
\usepackage{graphics}
\usepackage{epsfig}
\usepackage{arcs}
\usepackage{fontenc}
\usepackage{alltt}
\usepackage{amssymb,amsfonts,epsf,url}
\usepackage{graphicx}
\usepackage{pstricks,pstricks-add}
\usepackage{pst-node}
\usepackage{pst-coil}
\usepackage{amsmath}
\usepackage{amsthm}
\usepackage{verbatim}

\usepackage{fancyhdr}

\newtheorem{theorem}{Theorem}[section]
\newtheorem{lemma}[theorem]{Lemma}

\newtheorem{proposition}[theorem]{Proposition}

\newtheorem{definition}[theorem]{Definition}

\newtheorem{reduction}{Reduction Rule}[section]

\newlength{\alginputwidth}
\newlength{\algboxwidth}

\newsavebox{\algbox}
\newsavebox{\captionbox}
    {
        \setlength{\algboxwidth}{\columnwidth}
        \addtolength{\algboxwidth}{-\columnsep}
        \addtolength{\algboxwidth}{-1mm}
        \setlength{\alginputwidth}{\algboxwidth}
        \addtolength{\alginputwidth}{-1.7cm}
        \begin{figure}[tb]
            \vspace*{2mm}
            \centering
            \begin{lrbox}{\captionbox}
                \begin{minipage}[b]{\algboxwidth}
                    \centering
                    \caption{#1}
                    \label{#2}
                \end{minipage}
            \end{lrbox}
            \begin{lrbox}{\algbox}
                \begin{minipage}[b]{\algboxwidth}
                    \footnotesize
                    \vspace*{2mm}
    } 
    {
                    \vspace*{0.2mm}
               \end{minipage}
            \end{lrbox}
            \fbox{\usebox{\algbox}\hspace*{1mm}}
            \usebox{\captionbox}
            \vspace*{-4mm}
        \end{figure}
    }
\newsavebox{\algcodebox}
    {
        \begin{enumerate}
            \setlength{\itemsep}{2pt}
            \setlength{\parsep}{0pt}
            \setlength{\topsep}{0pt}
            \setlength{\parskip}{0pt}
            \setlength{\partopsep}{0pt}
    } 
    {\end{enumerate}}

\psset{unit=1pt}

\date{}
\title{What makes normalized weighted satisfiability tractable}

\author{{\sc Iyad Kanj}\thanks{School of
Computing, DePaul University, 243 S. Wabash Avenue, Chicago, IL
60604. Email: {\tt ikanj@cs.depaul.edu.} Phone: (+1) 312-362-5558.
Fax: (+1) 312-362-6116.}  \and {\sc Ge Xia}\thanks{Department of Computer
Science, Acopian Engineering Center, Lafayette College, Easton PA
18042, USA. Email: {\tt gexia@cs.lafayette.edu}. Phone: (+1)
610-330-5415. Fax: (+1) 610-330-5059.}}

\begin{document}

\maketitle

\begin{abstract}
We consider the weighted antimonotone and the weighted monotone satisfiability problems on normalized circuits of depth at most $t \geq 2$, abbreviated {\sc wsat$^-[t]$} and {\sc wsat$^+[t]$}, respectively. These problems model the weighted satisfiability of antimonotone and monotone propositional formulas (including weighted anitmonoone/monotone {\sc cnf-sat}) in a natural way, and serve as the canonical problems in the definition of the parameterized complexity hierarchy. We characterize the parameterized complexity of {\sc wsat$^-[t]$} and {\sc wsat$^+[t]$} with respect to the genus of the circuit. For {\sc wsat$^-[t]$}, which is $W[t]$-complete for odd $t$ and $W[t-1]$-complete for even $t$, the characterization is precise: We show that {\sc wsat$^-[t]$} is fixed-parameter tractable (FPT) if the genus of the circuit is $n^{o(1)}$ ($n$ is the number of the variables in the circuit), and that it has the same $W$-hardness as the general {\sc wsat$^-[t]$} problem (i.e., with no restriction on the genus) if the genus is $n^{O(1)}$. For {\sc wsat$^+[2]$} (i.e., weighted monotone {\sc cnf-sat}), which is $W[2]$-complete, the characterization is also precise: We show that {\sc wsat$^+[2]$} is FPT if the genus is $n^{o(1)}$ and $W[2]$-complete if the genus is $n^{O(1)}$. For {\sc wsat$^+[t]$} where $t > 2$, which is $W[t]$-complete for even $t$ and $W[t-1]$-complete for odd $t$, we show that it is FPT if the genus is $O(\sqrt{\log{n}})$, and that it has the same $W$-hardness as the general {\sc wsat$^+[t]$} problem if the genus is $n^{O(1)}$.
\end{abstract}

\section{Introduction} \label{sec:intro}
We consider the {\em weighted satisfiability} problems on monotone and antimonotone normalized circuits of depth at most $t \geq 2$. In the {\sc antimonotone weighted satisfiability} problem on normalized circuits of depth at most $t \geq 2$, abbreviated {\sc wsat$^-[t]$}, we are given a circuit $C$ of depth $t$ in the {\em normalized} form~\cite{fptbook,grohebook} (i.e., the output gate is an {\sc and}-gate, and the gates alternate between {\sc and}-gates and {\sc or}-gates) whose input literals are all negative,  and an integer parameter $k \geq 0$, and we need to decide if $C$ has a satisfying assignment of weight $k$. In the {\sc monotone weighted satisfiability} on normalized circuits of depth at most $t \geq 2$, abbreviated {\sc wsat$^+[t]$}, we are given a circuit $C$ of depth $t$ in the normalized form whose input literals are positive, and an integer parameter $k \geq 0$, and we need to decide if $C$ has a satisfying assignment of weight $k$. Our goal in this paper is to characterize the parameterized complexity of {\sc wsat$^-[t]$} and {\sc wsat$^+[t]$} ($t \geq 2$) with respect to the genus of circuit. We define the genus of the circuit to be the genus of the underlying undirected graph after the output gate is removed. This definition of the genus of the circuit is more general than the natural definition in which the genus is defined to be that of the whole circuit (output gate included) because an upper bound on the genus of the whole circuit implies the same upper bound on the genus of the circuit with the output gate removed. More specifically, all the results derived in the current paper, including the $W$-hardness results, hold true when the genus is defined to be that of the whole circuit. We mention that it is known that the {\sc weighted circuit satisfiability} problem on planar circuits with the output gate included of depth at most $t$ is solvable in polynomial time~\cite{chennn}. On the other hand, it can be shown via straightforward polynomial-time reductions from the ${\cal NP}$-hard problems {\sc planar vertex cover} and {\sc planar independent set}, that
{\sc wsat$^-[t]$} and {\sc wsat$^+[t]$} ($t \geq 2$) are ${\cal NP}$-complete on planar circuits (and hence on circuits of any genus) with the output gate removed. We also note that {\sc weighted circuit satisfiability} on planar circuits with unbounded depth is known to be $W[P]$-complete~\cite{abrahamson}.

The problems under consideration are of prime interest both theoretically and practically. From the theoretical perspective, they naturally represent the weighted satisfiability of (montone/antimontone) $t$-normalized propositional formulas, i.e., products-of-sums-of-products...(see, for example,~\cite{fptbook,grohebook}), including the canonical problems weighted antimonotone/monotone {\sc cnf-sat}. Moreover, the {\sc wsat$^-[t]$} and the {\sc wsat$^+[t]$} problems are used as the canonical complete problems for the different levels of the parameterized complexity hierarchy, the $W$-hierarchy, and the $W$-hierarchy can be defined based on them~\cite{fptbook,grohebook}. In particular, the {\sc wsat$^-[t]$} problem is $W[t]$-complete for odd $t \geq 3$, and {\sc wsat$^+[t]$} problem is $W[t]$-complete for even $t \geq 2$. Therefore, revealing the underlying structure that makes these problems (parameterized) tractable is important from the perspective of complexity theory. From a more practical perspective, {\sc wsat$^-[t]$} and {\sc wsat$^+[t]$} can be used to model several natural graph problems. Therefore, as mentioned in Section~\ref{sec:conclusion}, the results derived in the current paper can be used to obtain fixed-parameter tractability results for natural graph problems on graphs whose genus meets certain upper bounds by reducing these problems to {\sc wsat$^-[t]$} and {\sc wsat$^+[t]$}.

The computational complexity of many natural problems on planar graphs, and more generally on graphs whose genus meets certain upper bounds, have been extensively researched (see~\cite{meta,demaine,demaine1,fomin1,fomin2}, among others). In particular, it was shown that the bounded-genus property plays a key-role in determining the computational complexity (parameterized complexity including kernelization, subexponential-time computability, approximation) of a large class of graph problems. For example, using {\em bidimensionality theory}, it was shown in~\cite{demaine} that a large class of graph problems admit subexponential-time parameterized algorithm on graphs whose genus is upper bounded by a constant.
For graphs of larger genus (could be unbounded), it was shown in~\cite{genus} that the genus characterizes the computational complexity (parameterized complexity, approximation, subexponential-time computability) of some natural graph problems, including {\sc independent set} and {\sc dominating set}. For example, it was shown in~\cite{genus} that {\sc independent set} is FPT if the genus of the graph (on $n$ vertices) is $o(n^2)$, and is $W[1]$-complete if the genus is $\Omega(n^2)$. 

Research results on planar circuits, and on satisfiability problems defined on certain structures that are planar or that satisfy certain structural properties, are abundant.
Planar Boolean circuits have been extensively studied in the literature as they can be used to study VLSI chips, and they play an important role in deriving computational lower bounds for Boolean circuits~\cite{savage,turan,ingo}. After Lipton and Tarjan established their celebrated planar separator theorem, one of the first applications of the separator theorem they gave, was to derive lower bounds on the size of Boolean circuits that compute certain important functions~\cite{lipton}. The computational power of monotone planar circuits were also considered (e.g., see~\cite{monotone,mahajan}). Khanna and Motwani~\cite{KM96} studied the approximation of instances of satisfiability problems (weighted and unweighted) whose underlying structure is planar. More specifically, they studied satisfiability problems defined based on disjunctive normal form (DNF) formulas. The incidence graph of an instance of such problems is a bipartite graph that has a vertex for each variable and a vertex for each formula, and an edge between them if the variable occurs in the formula. They derived polynomial-time approximations schemes for instances of these problems whose underlying incidence graph is planar~\cite{KM96}. Cai et al.~\cite{limingapx} studied the parameterized complexity of the satisfiability problems introduced by Khanna and Motwani~\cite{KM96}, and showed that these problems are $W[1]$-hard even when the underlying incidence graph is planar. Researchers have also studied the parameterized complexity of {\sc cnf-sat} with respect to the treewidth of a graph defined based on the circuit (for example, see~\cite{szeider}).

In this paper, we characterize the parameterized complexity of {\sc wsat$^-[t]$} and {\sc wsat$^+[t]$} ($t \geq 2$) in terms of the genus of the circuit. For {\sc wsat$^-[t]$}, which is $W[t]$-complete for odd $t$ and $W[t-1]$-complete for even $t$, we give a tight characterization by showing that {\sc wsat$^-[t]$} is FPT if the genus of the circuit is $n^{o(1)}$ ($n$ is the number of the variables in the circuit), and that it has the same $W$-hardness as (the general) {\sc wsat$^-[t]$} if the genus is $n^{O(1)}$. The techniques used for deriving the FPT results for {\sc wsat$^-[t]$} can be summarized as follows. We first show how in {\em FPT-time} we can reduce an instance of {\sc wsat$^-[t]$} on circuits of genus $n^{o(1)}$ to an equivalent instance in which the number of occurrences of the literals is linear in $n$, and which has no {\em zero-variables}; we bound the number of occurrences using counting arguments that are based on Euler-type results for (multi) hypergraphs whose genus meets certain upper bounds. We then show that any instance of {\sc wsat$^-[t]$} in which the number of occurrences is linear and with no zero-variables admits a satisfying assignment whose weight is lower bounded by a function of $n$; this result is of independent interest. Combining the preceding two results, we conclude that the problem is FPT. For {\sc wsat$^+[t]$}, which is $W[t]$-complete for even $t$ and $W[t-1]$-complete for odd $t$, we give a tight characterization for $t=2$ (i.e., for weighted monotone {\sc cnf-sat}) by showing that {\sc wsat$^+[2]$} is FPT if the genus is $n^{o(1)}$ and $W[2]$-complete if the genus is $n^{O(1)}$. For $t  > 2$, we show that {\sc wsat$^+[t]$} is FPT if the genus is $O(\sqrt{\log{n}})$, and that it has the same $W$-hardness as {\sc wsat$^+[t]$} if the genus is $n^{O(1)}$. Both FPT results for $t = 2$ and $t > 2$ rely on a result showing that, for circuits of genus $n^{o(1)}$, there is a {\em Turing-fpt-reduction} that reduces an instance of {\sc wsat$^+[t]$} to {\em fpt-many} instances of the problem in which the number of gates that are incoming to the output gate of the circuit is a function of the parameter. Using this result, we can derive that {\sc wsat$^+[2]$} is FPT. For $t > 2$, we show that the aforementioned result implies that the treewidth of the circuit is $O(\log{n})$ if its genus is $O(\sqrt{\log{n}})$; this allows us to apply a dynamic programming approach to show that the problem  on genus $O(\sqrt{\log{n}})$ circuits is FPT. The hardness results for both {\sc wsat$^-[t]$} and {\sc wsat$^+[t]$} on circuits of genus $n^{O(1)}$ are derived by simple fpt-reductions from the general {\sc wsat$^-[t]$} and {\sc wsat$^+[t]$} problems.

Finally, we note that none of the algorithms presented in the current paper needs to  know in advance, nor needs it decide,
whether the minimum genus of the input circuit satisfies the required upper bounds or not.

\section{Preliminaries}
\label{sec:prelim}

\subsection{Graphs, hypergraphs, and genus}
\label{subsec:graphs}

We assume familiarity with the basic terminology and definitions in graph theory and parameterized complexity,
and refer the reader to~\cite{fptbook,grohebook,west} for more information.

Given an undirected graph $G$ and a vertex-set $S\subseteq V(G)$ such that the subgraph of $G$ induced by $S$, denoted $G[S]$, is connected, {\em contracting} $S$ in $G$ means removing all vertices in $S$ from $G$, and adding a new vertex that is adjacent to all former neighbors of the vertices in $S$ that are in $V(G) \setminus S$. For two adjacent vertices $u, v \in V(G)$, {\em contracting} the edge $uv$ in $E(G)$ means contracting the (connected) set of vertices $\{u, v\}$. Note that contracting an edge can result in a multigraph.

A {\em hypergraph} ${\cal H} = (V,E)$ consists of a {\em vertex set}
$V = V({\cal H})$ and an {\em edge set} $E = E({\cal H})$ so that $e
\subseteq V$ for every $e \in E$. If $E$ is allowed to be a multiset
(elements can repeat) we call ${\cal H}$ a {\em multihypergraph}. We also call the edges in a hypergraph {\em hyperedges}.

A graph has {\em genus $g$} if it can be drawn on a surface of genus
$g$ (a sphere with $g$ handles) without intersections. We say a
(multi)hypergraph ${\cal H}$ is {\em embeddable in a surface} if the
bipartite incidence graph obtained from ${\cal H}$ by replacing each
of its hyperedges by a vertex adjacent to all the vertices in the hyperedge is
embeddable in that surface. In particular, this definition allows us
to speak of a {\em planar (multi)hypergraph} or a {\em
(multi)hypergraph of genus $g$}. We refer the reader to~\cite{gross} for more information on the genus of a graph.

We have the following lemmas:

\begin{lemma}[\cite{im}]
\label{lem:edgebound}
A multihypergraph
of genus at most $g$ on $n$ vertices has at most $2n+4g-4$
 hyperedges containing at least three vertices, unless $n=1$ and $g=0$.
\end{lemma}

\begin{lemma}[Euler]\label{lem:Euler}
  A graph of genus $g$ on $n$ vertices contains at most $3n + 6g-6$ edges
  if $n\geq 3$.
\end{lemma}

\begin{lemma}[\cite{im}]\label{lem:im}
 A hypergraph of genus at most $g$  on $n$ vertices has at most $6n+10g-10$
 hyperedges if $n \geq 3$.
\end{lemma}

\subsection{Circuits, weighted satisfiability, and complexity functions}
\label{subsec:circuits}

A {\it circuit} is a directed acyclic graph. The vertices of
indegree $0$ are called the (input) {\it variables}, and are labeled either
by {\it positive literals} $x_i$ or by {\it negative literals}
$\overline{x}_i$. The vertices of indegree larger than $0$ are
called the {\it gates} and are labeled with Boolean operators
{\sc and} or {\sc or}. A special gate of outdegree $0$ is
designated as the {\it output} gate. We do not allow {\sc not}
gates in the above circuit model, since by De Morgan's
laws, a general circuit can be effectively converted into the
above circuit model. A circuit is said to be {\it monotone}
(resp. {\it antimonotone}) if all its input literals are positive
(resp. negative). The {\it depth} of a circuit is the maximum
distance from an input variable to the output gate of the circuit.
A circuit represents a Boolean function in a natural way. The size of a circuit $C$, denoted $|C|$, is the size of the underlying graph (i.e., number of vertices and edges).
An {\em occurrence} of a literal in $C$ is an edge from the literal to a gate in $C$. Therefore, the total number of occurrences of the literals in $C$ is the number of incoming edges from the literals in $C$ to its gates.

The {\em genus of a circuit} is the genus of the underlying undirected graph after the output gate has been removed.
We note that the definition that we use is more general than the natural definition, which defines the genus to be that of the whole circuit, i.e., including the output gate of the circuit (as explained in  Section~\ref{sec:intro}).

We consider circuits whose output gate is an {\sc and}-gate and that are in the {\em normalized} form (see~\cite{fptbook,grohebook}).
In the normalized form every (nonvariable) gate has outdegree at most 1, and starting from the output {\sc and}-gate, the gates are structured into alternating levels of {\sc or}s-of-{\sc and}s-of-{\sc or}s... We denote a circuit that is in the normalized form and that is of depth at most $t \geq 2$ by a $\Pi_t$ circuit. We write $\Pi_t^+$ to denote a monotone $\Pi_t$ circuit, and $\Pi_t^-$ to denote an antimonotone $\Pi_t$ circuit. We do not assume that the literals appear at the same (top) level of the circuit. $\Pi_t$ circuits naturally represent the satisfiability of $t$-normalized propositional formulas; that is, formulas that are products-of-sums-of-products...(see, for example,~\cite{fptbook,grohebook}), including the canonical problem {\sc cnf-sat}, which is complete for the class ${\cal NP}$.

Throughout this paper, we implicitly assume that the following simplifications are performed always (i.e., as soon as one of them applies).
The first simplification takes place when there exist two gates of the same type (i.e., both are {\sc or}-gates or both are {\sc and}-gates) such that one is incoming to the other. In this case the two gates are {\em merged} into a new gate of the same type (i.e., the edge between them is contracted and possible multiple edges are removed); note that this reduction does not increase the genus of the circuit, even if one of the two gates is the output gate of the circuit. The second simplification takes place when there is a gate $g$ of indegree 1 in $C$. In this case we connect the input of $g$ to the gate that $g$ is incoming to (i.e., contract the edge between $g$ and the gate that $g$ is incoming to), and remove $g$. Again, note that this simplification does not increase the genus of the circuit. We will assume at every point that: every gate with outdegree 0 except the output gate is removed, every gate has indegree at least 2, and that no two gates of the same type such that one is incoming to the other exist. Note that the resulting circuit from the aforemention simplifications is equivalent to the original circuit.

We say that a truth assignment $\tau$ to the variables
of a circuit $C$ {\it satisfies} a gate $g$ in $C$ if $\tau$ makes the gate $g$ have
value $1$, and that $\tau$ {\it satisfies the circuit $C$} if $\tau$
satisfies the output gate of $C$. A circuit $C$ is {\it satisfiable}
if there is a truth assignment to the input variables of $C$ that
satisfies $C$. The {\it weight} of an assignment $\tau$ is the number
of variables assigned value $1$ by $\tau$. An indegree-2 gate is called a {\em 2-literal gate} if both its incoming edges are from literals.
A {\em critical gate} in a $\Pi_t$ circuit $C$ is an {\sc or}-gate that is connected to the output {\sc and}-gate
of the circuit; clearly, any satisfying assignment to $C$ must satisfy all critical gates in $C$. If we remove the literals from $C$, we obtain a directed graph whose underlying  undirected graph is a tree $T_C$. If we root $T_C$ at the output gate of $C$, we can now use the terms {\em child(ren)}, {\em parent}, {\em grandparent} of a gate in $T_C$ in a natural way. Note that every literal in $C$ is connected to some gates in $T_C$. For a gate $g$ in $T_C$, we denote by $T_g$ the subtree of $T_C$ rooted at $g$. We may regard an edge in $T_C$ between a child $g'$ of a gate $g$ and $g$, or between a literal and gate $g$, as an incoming edge to $g$.

A {\it parameterized problem} $Q$ is a subset of $\Omega^* \times
\mathbb{N}$, where $\Omega$ is a fixed alphabet and $\mathbb{N}$
is the set of all non-negative integers. Each instance of the
parameterized problem $Q$ is a pair $(x, k)$, where the second
component, i.e., the non-negative integer $k$, is called the {\it
parameter}. We say that the parameterized problem $Q$ is
{\it fixed-parameter tractable} \cite{fptbook}, shortly $FPT$, if there is a
(parameterized) algorithm that decides whether an input $(x, k)$
is a member of $Q$ in time $f(k)|x|^{O(1)}$, where $f(k)$ is a computable function independent of the
input length $|x|$.  Let $FPT$ denote the class of all fixed-parameter
tractable parameterized problems. (We abused the notation ``$FPT$" above for simplicity.) A parameterized problem $Q$
is {\it fpt-reducible} to a parameterized problem $Q'$ if there is
an algorithm that transforms each instance $(x, k)$ of $Q$
into an instance $(x', g(k))$ ($g$ is a function of $k$ only) of
$Q'$ in time $f(k)|x|^{O(1)}$, where $f$ and $g$ are computable
functions of $k$, such that $(x, k) \in Q$ if and
only if $(x', g(k)) \in Q'$. Based on the notion of fpt-reducibility, a hierarchy of
fixed-parameter intractability, {\it the $W$-hierarchy} $\bigcup_{t
\geq 0} W[t]$, where $W[t] \subseteq W[t+1]$ for all $t \geq 0$, has
been introduced, in which the $0$-th level $W[0]$ is the class {\it
$FPT$}. The hardness and completeness have been defined for each level
$W[i]$ of the $W$-hierarchy for $i \geq 1$ \cite{fptbook}. It is
commonly believed that $W[1] \neq FPT$ (see \cite{fptbook}). The
$W[1]$-hardness has served as the hypothesis for fixed-parameter
intractability.

For $t \geq 2$, the {\sc weighted $\Pi_t$-circuit satisfiability} problem,
abbreviated {\sc wsat$[t]$} is for a given $\Pi_t$-circuit $C$ and
a given parameter $k$, to decide if $C$ has a satisfying assignment
of weight $k$. The {\sc weighted monotone $\Pi_t$-circuit
satisfiability} problem, abbreviated {\sc wsat$^+[t]$}, and the
{\sc weighted antimonotone $\Pi_t$-circuit satisfiability} problem,
abbreviated {\sc wsat$^-[t]$} are the {\sc wsat$[t]$} problems on
monotone circuits and antimonotone circuits, respectively. We denote
by {\sc wsat$^-$} the {\sc wsat$^-[2]$} problem, and by {\sc wsat$^+$} the {\sc wsat$^+[2]$} problem (i.e., the weighted antimonotone/monotone {\sc cnf-sat} problem).
It is known that for each even integer $t \geq 2$,
{\sc wsat$^+[t]$} is $W[t]$-complete, and for each odd integer $t
\geq 2$, {\sc wsat$^-[t]$} is $W[t]$-complete; moreover, {\sc wsat$^-$}  is $W[1]$-complete~\cite{fptbook,grohebook}.

The (time) complexity functions used in this paper are assumed to be proper complexity functions that are unbounded and nondecreasing.
For a complexity function $f: \mathbb{N} \rightarrow \mathbb{N}$, we define its inverse, $f^{-1}$, by $f^{-1}(h) = \max\{q \mid f(q)
\leq h\}$. Since the function $f$ is nondecreasing and unbounded,
the function $f^{-1}$ is also nondecreasing and unbounded, and
satisfies $f(f^{-1}(h)) \leq h$. We shall also assume that the complexity functions and their inverses
can be computed efficiently (i.e., in time linear in the input size and the value of the function).
The $o(\cdot)$ notation used in this paper denotes the $o^{\mbox{eff}}(\cdot)$ notation (see,
for instance, \cite{grohebook}). More formally, for any two computable functions $f, g: \mathbb{N} \rightarrow \mathbb{N}$, by writing
$f(n) = o(g(n))$ we mean that there exists a computable nondecreasing
unbounded function $\mu(n) : \mathbb{N} \rightarrow \mathbb{N}$,
and $n_0 \in \mathbb{N}$, such that $f(n) \leq g(n)/\mu(n)$ for
all $n \geq n_0$.

By {\em fpt-time}, we denote time complexity of the form $f(k)N^{O(1)}$, where $N$ is the input length,
and $k$ is the parameter, and $f$ is a complexity function of $k$.

The following lemma is folklore:

\begin{lemma}
\label{lem:fpttime}
The two functions $N^{o(1)h(k)}$ and $(\log{N})^{h(k)}$ ($f, h$ are complexity functions) are bounded above by $f(k)N^{O(1)}$. Therefore, if a parameterized problem is solvable in time that is upper bounded by either of these two functions, where $N$ is the input length and $k$ is the parameter, then the problem is
solvable in fpt-time, and hence is $FPT$.
\end{lemma}

\begin{proof} (Sketch)
Suppose that the parameterized problem is solvable in time $N^{f(k)/\mu(N)}$, for some complexity function $\mu(N)$. By considering the two cases $f(k) \leq \mu(N)$ and $f(k) > \mu(N)$ (and hence, $k > \mu^{-1}(N)$),
it can be shown using a folklore argument that the problem is $FPT$. The proof for the other function is similar.
\end{proof}

\section{The antimonotone case}
\label{sec:antimonotone}

In this section we give a complete characterization of the parameterized complexity of the {\sc wsat$^-[t]$} problem ($t \geq 2$) with respect to the genus of the circuit. We start with the following hardness result:

\begin{theorem}\label{lem:hardnessantimonotone}
Let $c > 0$ be a constant. The {\sc wsat$^-[t]$} ($t \geq 2$) problem on circuits of genus $g(n) =\Omega(n^c)$, where $n$ is the number of variables in the circuit, is $W[t]$-complete for odd $t$ and $W[t-1]$-complete for even $t$.
\end{theorem}

\begin{proof}
To prove the hardness result in the theorem, we show that {\sc wsat$^-[t]$} is fpt-reducible to {\sc wsat$^-[t]$} on circuits of genus $g(n) =\Omega(n^c)$. Since {\sc wsat$^-[t]$} is $W[t]$-hard for odd $t$, and $W[t-1]$-hard for even $t$, the hardness result follows. Suppose that $g(n) = c'n^c$, for some constant $c' > 0$.

Let $(C_0, k)$ be an instance of {\sc wsat$^-[t]$}, where $C_0$ is a $\Pi_{t}^-$ circuit and $k$ is the parameter. Suppose that $C_0$ has $n_0$ variables and $m_0$ gates (including the variables). Therefore, the genus of $C_0$ is at most $m_0^2$. If $m_0^2 \leq c'n_0^c$, then the fpt-reduction outputs the instance $(C, k)$, where $C= C_0$. If $m_0^2 > c'n_0^c$, let $C$ be the circuit obtained from $C_0$ by adding $\lceil(m_0^2/c')^{(1/c)}\rceil - n_0$ new negative literals that are incoming to the output {\sc and}-gate of $C_0$. The fpt-reduction outputs the instance $(C, k)$. Obviously, the genus of $C$ is at most that of $C_0$, which is at most $m_0^2$. It can be easily verified that the genus of $C$, in both cases, is at most $c'n^c$, where $n$ is the number of variables in $C$. Noting that the new literals (if added) must be assigned value 1, and hence their corresponding variables value 0, by any satisfying assignment of $C$, we conclude that $C_0$ has a weight-$k$ satisfying assignment if and only if $C$ has a weight-$k$ satisfying assignment. It follows that the above reduction is an fpt-reduction from {\sc wsat$^-[t]$} to {\sc wsat$^-[t]$} on circuits of genus $g(n) =\Omega(n^c)$.

The completeness of the problem follows from the membership of {\sc wsat$^-[t]$} in $W[t]$ for odd $t$, and in $W[t-1]$ for even $t$.
\end{proof}

\begin{definition}\rm
Let $C$ be a $\Pi_t^-$ circuit, and let $x_i$ be a variable in $C$. We say that $x_i$ is a {\em zero-variable} for $C$ if assigning $x_i = 1$ causes $C$ to evaluate to 0. Therefore, any zero-variable for $C$ must be assigned the Boolean value 0 in a satisfying truth assignment for $C$. A {\em nonzero-variable} for $C$ is a variable that is not a zero-variable for $C$. A $\Pi_t^-$ circuit $C$ {\em has no zero-variables} if all the variables in $C$ are nonzero-variables.
\end{definition}

We note that determining whether a variable $x_i$ is a zero-variable for a $\Pi_t^-$ circuit $C$ can be done in polynomial time.

\begin{proposition}
\label{prop:removezeros}
Let $(C, k)$ be an instance of {\sc wsat$^-[t]$} ($t \geq 2)$ such that the genus of $C$ is $g(n)=n^{o(1)}$. In fpt-time, we can either solve $(C, k)$, or reduce it to an equivalent instance $(C', k)$ where $C'$ has genus at most $g(n)$ and no zero-variables, and such that the number of variables $n'$ in $C'$ satisfies $g(n) \leq n' \leq n$.
\end{proposition}

\begin{proof}
Observe that if $(C, k)$ has a satisfying assignment of weight $k$, then none of the variables assigned 1 by such an assignment can be a zero-variable of $C$.

Suppose first that the number of nonzero-variables in $C$ is $n^{o(1)}$, and let ${\cal N}$ be the set of nonzero-variables of $C$. We enumerate each subset $S$ of ${\cal N}$ of size $k$ as a candidate subset of variables that will be assigned 1 by a satisfying assignment of weight $k$ for $C$. For each such candidate subset $S$, we assign the variables in $S$ the value 1 and the remaining variables in $C$ the value 0, and check if the assignment satisfies $C$; if it does, we accept $(C, k)$. If no enumerated subset leads to acceptance, we reject $(C, k)$. The number of the enumerated subsets is ${|{\cal N}| \choose i} = n^{o(1)k}$. By Lemma~\ref{lem:fpttime}, the above algorithm runs in fpt-time.

We may now assume that the number of nonzero-variables in $C$, $n'$, is at least $g(n)=n^{o(1)}$. Let $C'$ be the circuit obtained from $C$ by assigning the zero-variables of $C$ the value 0. Observe that this assignment does not introduce zero-variables, and hence the resulting circuit $C'$ has no zero-variables, and satisfies the statement of the lemma.
\end{proof}

Let $v$ and $v'$ be vertices in $C$. We say that $v$ and $v'$ are {\em equivalent} if $v$ and $v'$ are literals and $v=v'$, or both $v$ and $v'$ are 2-literal gates that are of the same type (either both are {\sc and}-gates or both are {\sc or}-gates) and have the same two literals incoming to them.

We apply the following reduction rule repeatedly until it is not applicable:

\begin{reduction}\rm
\label{red:1}
Let $C$ be $\Pi_t^-$ circuit, and let $g$ be a gate in $C$. Let $v$ be a literal or a 2-literal gate that is incoming to $g$.

\begin{itemize}
\item[(a)] If there exists a vertex $v'\neq v$ that is equivalent to $v$, such that $v'$ is incoming to $g$, then let $C'$ be the circuit resulting from $C$ after removing the edge from $v'$ to $g$.

\item[(b)] If $g$ is an {\sc or}-gate and there exists a gate $g' \neq g$ in the subtree $T_{g}$ of $T_C$ and a vertex $v'$ equivalent to $v$ such that $v'$ is incoming to $g'$, then let $C'$ be the circuit resulting from $C$ after performing the following: if $g'$ is an {\sc and}-gate then remove $g'$, and if $g'$ is an {\sc or}-gate then remove the edge from $v'$ to $g'$.

\item[(c)] If $g$ is an {\sc and}-gate and there exists a gate $g' \neq g$ in $T_{g}$ and a vertex $v'$ equivalent to $v$ such that $v'$ is incoming to $g'$, then let $C'$ be the circuit resulting from $C$ after performing the following: if $g'$ is an {\sc or}-gate then remove $g'$, and if $g'$ is an {\sc and}-gate then remove the edge from $v'$ and $g'$.
\end{itemize}

The circuit $C'$ is a $\Pi_t^-$ circuit that is equivalent to $C$.
\end{reduction}

\begin{proof}
We prove the correctness for the case when $v = \overline{x}_{j}$ is a literal. The proof is very similar for the case when $v$ is a 2-literal gate.

It suffices to show that any truth assignment $\tau$ satisfies $C$ if and only if it satisfies $C'$. Since the only differences between $C$ and $C'$ occur in $T_{g}$ (including the literals connected to the gates in $T_{g}$),
it suffices to show that the value of $g$ induced by $\tau$ in $C$ is the same as that in $C'$. This is clear for part (a), so we prove it for part (b), and the proof for (c) is similar. Note that, by the simplification rules, we can assume that every gate has indegree at least 2.

If $\overline{x}_{j}$ is assigned 1 by $\tau$, then clearly the value of $g$ induced by $\tau$ in both $C$ and $C'$ is 1, and hence is the same. Suppose now that $\overline{x}_{j}$ is assigned 0 by $\tau$. An {\sc and}-gate in $T_{g}$ that $\overline{x}_{j}$ is incoming to evaluates to 0 by $\tau$, and hence its removal from $C$ does not affect the value of $g$ induced by $\tau$; similarly, since $\overline{x}_{j}=0$, the value of an {\sc or}-gate in $T_{g}$ to which $\overline{x}_{j}$ is incoming, is not affected by the removal of the connection from $\overline{x}_{j}$ to this gate, and hence this removal does not affect the value of $g$ induced by $\tau$. It follows that the value of $g$ induced by $\tau$ is the same in both $C$ and $C'$.
\end{proof}

Note that all the simplification rules and the reduction rule do not increase the genus of $C$, nor do they decrease the number of variables/lietrals in $C$. Moreover, these operations can be carried out in time polynomial in the size of the circuit.

\begin{lemma}
\label{lem:occurrences}
Let $C$ be a $\Pi_t^-$ circuit on $n$ variables of genus $g(n) \leq n$ such that $C$ has no zero-variables. In polynomial time we can reduce $C$ to an equivalent $\Pi_t^-$ circuit $C'$ of genus $g(n)$ on the same set of variables such that the number of occurrences of the literals in $C'$ is $O(n)$.
\end{lemma}

\begin{proof}
We apply Reduction Rule~\ref{red:1} to $C$ until it is no longer applicable. (We also assume that the simplification rules are applied as discussed before.) Let $C'$ be the resulting circuit. From the above reduction rule, we know that $C'$ is equivalent to $C$, and hence $C'$ has no zero-variables. Since none of these rules remove any variables/literals, $C'$ has the same variables as $C$. Moreover, all operations performed by the reduction and simplification rules either remove edges, gates, or are edge contractions. Therefore, the genus of $C'$ is at most $g(n)$. It remains to show that the number of occurrences of the literals in $C'$ is $O(n)$.

To simplify the counting, we divide the occurrences of the literals in $C'$ into three types: (1) occurrences of literals incoming to a gate $g$ such that $g$ has degree at least 3 and all incoming edges to $g$ are from literals; (2) occurrences of literals incoming to 2-literal gates; and (3) all other occurrences, which are the occurrences of literals incoming to a gate that has at least one gate incoming to it. Next, we upper bound the number of occurrences of each type. Note that since $C'$ has no zero-variables, no literal is incoming to the output gate of $C'$. Let $C'^-$ be $C'$ with the output gate removed.

To bound the number of type-(1) occurrences, we define the multihypergraph ${\cal H}$ whose vertex-set is the set of literals/variables in $C'$. Call a gate $g$ of degree at least 3 whose incoming edges are all from literals, a {\em type-(1) gate}. For each type-(1) gate $g$, we correspond a hyperedge in $H$ that contains the literals that are incoming to $g$. Clearly, the number of occurrences of the literals that are incoming to the type-(1) gates is the same as the total number of occurrences of the vertices of ${\cal H}$ in its hyperedges. Since the genus of $C'^-$ is at most $g(n)$, by the definition of the genus of a hypergraph, the genus of ${\cal H}$ is at most $g(n)$ since its incidence graph is a subgraph of the underlying graph of $C'^-$.  Since each hyperedge in ${\cal H}$ has size at least 3, by Lemma~\ref{lem:edgebound}, the number of hyperedges in the multihypergraph ${\cal H}$ is $O(n + g(n))=O(n)$. Therefore, the incidence graph ${\cal I}$ of ${\cal H}$ has $O(n)$ vertices and genus $g(n)$. By Lemma~\ref{lem:Euler}, the number of edges in ${\cal I}$, which is the same as the total number of vertices in the hyperedges of ${\cal H}$, is $O(n)$. This shows that the number of type-(1) occurrences is $O(n)$.

To upper bound the number of type-(2) occurrences, we upper bound the number of 2-literal gates. First, consider the set ${\cal G}_0$ of 2-literal gates that are incoming to the output gate of $C'$, and ignore all other gates for now. We start by upper bounding the cardinality of ${\cal G}_0$. Since all gates in ${\cal G}_0$ are incoming to the output gate of $C'$, by Reduction Rule~\ref{red:1}, and since all gates in ${\cal G}_0$ are {\sc or}-gates, any pair of literals in $C$ can be incoming to at most one gate in ${\cal G}_0$.  Therefore, we can define a (simple) graph whose vertex-set is the set of literals in $C'$, and whose edges correspond to the gates in ${\cal G}_0$. Clearly, the genus of the constructed graph is $g(n)$. By Lemma~\ref{lem:Euler}, the number of edges in this graph, which is the same as the number of gates in ${\cal G}_0$, is $O(n)$. It follows that the cardinality of ${\cal G}_0$ is $O(n)$, and hence, the number of type-(2) occurrences that are incoming to gates in ${\cal G}_0$ is $O(n)$. Now we upper bound the number of 2-literal gates that are not in ${\cal G}_0$; let ${\cal G}_1$ be the set of these gates. First, we upper bound the number of critical gates in $C'$ that are not in ${\cal G}_0$ by $O(n)$. To do so, observe that each such critical gate $g$ has at least three literals incoming to the gates in $T_g$ (note that there are no gates of indegree 1). By contracting the edges in $T_g$, for each critical gate $g$, and removing any resulting multiple edges, we obtain a vertex that is connected to at least three distinct literals in $C'$; the fact that the resulting vertex is connected to at least three distinct literals follows from the simplification rules and from Reduction Rule~\ref{red:1}, and can be easily verified by the reader. We correspond this resulting vertex with gate $g$. Now by defining a multihypergraph whose vertices are the literals in $C'$, and whose hyperedges correspond to the vertices resulting from the contractions, we can upper bound the number of such vertices, and hence the number of critical gates in $C'$ by $O(n)$, in a similar fashion to that of bounding the type-(1) gates above. (Note that the genus of the defined multihypergraph is at mots $g(n)$ since its incidence graph is a subgraph of a contraction of $C'^-$.) To upper bound the number of gates in ${\cal G}_1$, apply the following operation until it is no longer applicable: For each gate $g$ in ${\cal G}_1$, if $g$ is not incoming to a critical gate, contract the edge between the parent of $g$ in $T_{C'}$ and the grandparent of $g$ in $T_{C'}$. After the application of the aforementioned operation, each gate in ${\cal G}_1$ is incoming to a critical gate, and has exactly two literals incoming to it. Now define a multihypergraph whose
vertex-set consists of the set of literals in $C'$ plus the critical gates, and whose hyperedges contain the vertices that the gates in ${\cal G}_1$ are adjacent to after these contractions; note that each hyperedge in this multihypergraph has size at least 3. Clearly, the defined multihypergraph has genus $g(n)$ since it is a contraction of a subgraph of $C'^-$. Since the number of critical gates in $C'$ is $O(n)$, it follows from Lemma~\ref{lem:edgebound} that the number of gates in ${\cal G}_1$ is $O(n)$. Summing up, the number of 2-literal gates in $C'$ is $O(n)$, and hence the number of type-(2) occurrences is $O(n)$.

Finally, to upper bound the type-(3) occurrences, we again define a multihypergraph ${\cal H}$ of genus $g(n)$ whose vertex-set is the set of literals in $C'$,
and use a charging scheme to charge the type-(3) occurrences to the total number of occurrences of the vertices of ${\cal H}$ in its hyperedges. To ensure that the genus of ${\cal H}$ is $g(n)$, we
rely on the forest ${\cal F}$ in $C'^-$, resulting from $T_{C'}$ after removing the output gate of $C'$, when defining ${\cal H}$. Call a gate a {\em type-(3) gate}  if it has a type-(3) literal incoming to it. We define the {\em level} of a gate to be the distance from it to the output gate of $C'$. We start the charging argument at the type-(3) gates at the highest level of the circuit, and go from the top to the bottom (we assume that the output gate is at the bottom of the circuit). Since $C'$ has no zero-variables, no type-(3) occurrence is incoming to the output gate of $C'$, and hence this charging scheme will stop at the critical gates of $C'$. Consider a type-(3) gate $g$ at the highest level. Since $g$ is not a type-(2) gate and its indegree is more than 1, the number of distinct literals incoming to the subtree $T_g$ in ${\cal F}$ is at least 3. Note that any literal that is incoming to $g$ is not incoming to any other gate in $T_g$ by Reduction Rule~\ref{red:1}. Therefore, by contracting $T_g$ to a single vertex and removing any resulting multiple edges, we get a vertex that is adjacent to all the literals that are incoming to $T_g$, including the type-(3) literals incoming to $g$, and such that the degree of this vertex is at least 3. We associate a hyperedge in ${\cal H}$ with the vertex resulting from this contraction that contains the literals incoming to the resulting vertex. Note that each type-(3) occurrence that is incoming to $g$ corresponds to a literal contained in the created hyperedge. In particular, since each type-(3) literal incoming to $g$ is not incoming to any other gate in $T_g$, no multiple edge that was removed corresponds to any such type-(3) literal, and all type-(3) literals incoming to $g$ are accounted for by (i.e., charged to) the corresponding literals in the defined hyperedge. Consider now a type-(3) gate $g$, and assume inductively, that we finished processing all type-(3) gates above it. We can assume that $g$ has at least one type-(3) gate above it; otherwise, the treatment is similar to that of the base case. If more than one type-(3) gate in $T_{g}$ have been charged in the above scheme, we keep one of them, and remove the edges between each other gate and its parent in $T_g$, thus disconnecting the (contracted) vertex corresponding to the gate from ${\cal F}$; after this process, exactly one type-(3) gate in the resulting $T_g$ was charged earlier in the charging scheme. Again, note that no type-(3) literal that is incoming to $g$ can be incoming to any gate in $T_{g}$. Now we contract the edges in $T_g$ and remove any resulting multiple edges to form a hyperedge of size at least 3 that contains all type-(3) occurrences incoming to $g$ (this can be viewed as if we are adding the type-(3) literals incoming to $g$ to the hyperedge corresponding to the single charged type-(3) gate in $T_{g}$). This charging scheme stops at the critical gates of $C'$. At that point, we have defined a multihypergraph ${\cal H}$ whose genus is $g(n)$ since all the hyperedges in ${\cal H}$ were defined based on contractions of subtrees in ${\cal F}$. The total number of type-(3) occurrences in $C'$ is at most the the total number of occurrences of the vertices of ${\cal H}$ in its hyperedges. Using a similar argument to that used for upper bounding the number of type-(1) occurrences, we conclude that the number of type-(3) occurrences in $C'$ is $O(n)$.

It follows that the total number of occurrences of the literals in $C'$ is $O(n)$. This completes the proof.
\end{proof}

\begin{theorem}
\label{thm:sparse}
Let $C$ be $\Pi_t^-$ circuit with $n$ variables such that $C$ has no zero-variables and the number of occurrences of the literals in $C$ is $O(n)$.  $C$ has a satisfying assignment in which at least $f(n) = \log^{(d^{t})}n$ variables are assigned 1, where $\log^{(i)}$ indicates the logarithm (base 2) applied $i$ times, and $d > 0$ is an integer constant whose value is to be fixed in the proof.
\end{theorem}

\begin{proof}

We will prove the statement of the lemma by induction on $t$.

Since the number of occurrences of the literals in $C$ is $O(n)$, without loss of generality, assume that the degree of every variable (or equivalently literal) in $C$ is at most a constant $d > 0$. If this is not the case, we can assign 0 to all the variables of degree more than $d$ (satisfied gates are then removed), and there are $\Omega(n)$ remaining variables, each of degree at most $d$. (We can redefine $n$ and $d$ if necessary). Therefore, the number of occurrences of the literals in $C$ is at most $dn$.

We say that a gate or a literal, $g$, {\em contains a variable} $v$ if there is a path from the literal $\overline{v}$ to $g$. Denote by $V(g)$ the set of variables contained in $g$ (if $g$ is a literal, then $g=\overline{v}$, and $V(g) = \{v\}$). Let $v$ be a variable contained in a gate $g$. We call $v$ a {\em zero-variable} for $g$ if assigning $v$ the value 1 falsifies $g$; otherwise, $v$ is called a {\em nonzero-variable} for $g$. In particular, a zero-variable (resp. nonzero-variable) for the output gate of $C$ is a zero-variable (resp. nonzero-variable) for $C$, as previously defined.

When $t=2$, every {\sc or}-gate incoming to the output gate of $C$ contains at least two literals. Keep only two literals for each such {\sc or}-gate, and remove the edges from the other literals to the {\sc or}-gate (without removing the literals from $C$). The problem becomes the {\sc independent set} problem on multigraphs of degree bounded by $d$, which can be easily seen to have a solution of size $\Omega(n)$. By assigning 1 to the variables in the independent set and 0 to the remaining variables, the circuit is satisfied. The statement follows.

For simplicity of the presentation and to avoid repetition, the proof of the other base case when $t=3$ (we induct on $t-2$) will be combined with the proof of the inductive step, with the understanding that when $t=3$ the inductive hypothesis does not apply, as explained later in the proof. Assume in what follows that $t \geq 3$, and that the statement is true for any circuit of depth smaller than $t$ that satisfies the statement of the lemma.

First, observe that in the case when $d=1$, $C$ has a satisfying assignment in which at least $n/2$ variables are assigned 1. This can be seen as follows. Let $g_1, g_2, \ldots, g_r$ be the vertices incoming to the output gate
of $C$. Since $C$ has no zero-variables, each $g_i$, for $1\leq i\leq r$, must be an {\sc or}-gate having at least two vertices incoming to it; we use a vertex here to denote a gate or a literal. From each $g_i$, pick a vertex $v_i$ incoming to it that contains at most half of the variables contained in $g_i$; this can be done since every literal in $C$ occurs exactly once. By assigning all variables in $v_i$, for $i=1, \ldots, r$, the value 0, and all the remaining variables in $C$ the value 1, we obtain an assignment that satisfies $C$, and in which at least half of the variables are assigned 1.

Suppose now that $d \geq 2$. Consider the following procedure:

Fix a variable in $C$; without loss of generality, let it be $x_1$ and let $g_1, g_2, \ldots, g_l$, where $l\leq d$, be the {\sc or}-gates incoming to the output {\sc and}-gate of $C$ that contain $x_1$. For an arbitrary $g_i$, $1 \leq i \leq l$, if assigning $x_1$ the value 1 falsifies $g_i$ then $x_1$ would be a zero-variable for the circuit, which is not possible. Therefore, there must exist an {\sc and}-gate or a literal, denoted $w^1_i$, incoming to $g_i$ that is not falsified by assigning $x_1$ the value 1. Let $U$ be the set of variables consisting of $x_1$ plus all the variables contained in $w^1_1, w^1_2, \ldots, w^1_l$. Consider the following cases: \\

\noindent{\bf Case 1.} \  If $|U| \leq nd/(d+1)$, then assign $x_1$ the value 1, and the other variables in $U$ the value 0. Every $w^1_i$, and every hence $g_i$, for $i=1, \ldots, l$, is satisfied by this assignment. Afterwards, every $g_i$ can be removed, and the resulting circuit has at least $n-nd/(d+1) = n/(d+1)$ variables left.

\noindent{\bf Case 2.} \ If $|U| > nd/(d+1)$, then one of $w^1_1, w^1_2, \ldots, w^1_l$ contains at least $\frac{nd/(d+1)}{l} \geq \frac{nd/(d+1)}{d} = n/(d+1)$ variables in $U$; without loss of generality, let $w^1_1$ be such a one. We further distinguish the following subcases:

\begin{itemize}
\item[2.1.] If at most half of the variables of $w^1_1$ are zero-variables of $w^1_1$ (note that this case does not apply when $t=3$, because when $t=3$ all variables contained in $w^1_1$ are zero-variables of $w^1_1$), then assign the zero-variables of $w^1_1$ the value 0. Afterwards, $w^1_1$ is a $(t-2)$-level circuit of at least $n/(2d+2)$ nonzero-variables. Applying the inductive hypothesis to $w^1_1$, we know that $w^1_1$ has a satisfying assignment with at least $\log^{(d^{t-2})}(\frac{n}{2d+2})$ variables assigned 1. This means that in the antimonotone circuit, if we assign 0 to all but these variables, $w^1_1$ is satisfied and so is $g_1$, which can then be removed. Now the resulting circuit $C$ has at least $\log^{(d^{t-2})}(\frac{n}{2d+2})$ variables left, whose degree is at most $d-1$ because they are all incoming to gates in $T_{g_1}$, which is removed.

\item[2.2.]  If any of the {\sc and}-gates or literals incoming to $g_1$, say $w^2_1$, shares fewer than $n/(2d+2)$ variables with $w^1_1$, then $|V(w^1_1) \setminus V(w^2_1)| \geq n/(2d+2)$. Assigning 0 to all variables in $V(w^2_1)$ will satisfy $w^2_1$ and hence will satisfy $g_1$, which can then removed. So the circuit $C$ will have at least $n/(2d+2)$ variables (in $V(w^1_1) \setminus V(w^2_1)$), whose degree is at most $d-1$ (because $g_1$ is satisfied and removed).

\item[2.3.] Now assume that each {\sc and}-gate incoming to $g_1$ shares at least $n/(2d+2)$ variables with $w^1_1$, and hence each contains at least $n/(2d+2)$ variables. Since the total number of occurrences of the literals in $C$ is at most $dn$, there are at most $\frac{dn}{n/(2d+2)} = 2d(d+1)$ {\sc and}-gates incoming to $g_1$. Let $\gamma$ be the number of variables such that each is a nonzero-variable for at least one {\sc and}-gate incoming to $g_1$. We distinguish two subcases:

\begin{itemize}
\item[2.3.1.] If $\gamma \geq n/(2d+2)$, then there exists an {\sc and}-gate incoming to $g_1$, denoted by $w'$, that contains at least $\frac{n/(2d+2)}{2d(d+1)} = \frac{n}{4d(d+1)^2}$ nonzero-variables. (Note that this case does not apply when $t=3$, when every variable is a zero-variable for every {\sc and}-gate incoming to $g_1$ that the variable is contained in.) By a similar argument to that made in 2.1, we apply the inductive hypothesis to $w'$. Afterwards, the circuit $C$ has at least $\log^{(d^{t-2})}(\frac{n}{4d(d+1)^2})$ variables, whose degree is at most $d-1$.

\item[2.3.2.] \ If $\gamma < n/(2d+2)$, assign 0 to every nonzero-variable contained in a gate that is incoming to $g_1$. The remaining variables of $g_1$ are zero-variables of the {\sc and}-gates (or literals) incoming to $g_1$. In other words, what results of $g_1$ is an {\sc or}-gate of the form: $w^1_1 \vee w^2_1 \vee \ldots \vee w^s_1$, where $s\leq 2d(d+1)$ and each $w^j_1$ is a literal or an {\sc and}-gate whose incoming edges are all from literals. Note that there are at most $2d(d+1)$ {\sc and}-gates (or literals) in $g_1$ and $w^1_1$ has at least $n/(2d+2)$ variables left. Denote by $U_j$ be the set of variables shared by all $w^1_1, \ldots, w^j_1$: $$U_j = V(w^1_1) \cap \ldots \cap V(w^j_1).$$ Consider the following process:

    If $|U_2| \leq |U_1|/2$, then $|U_1\setminus U_2| \geq |U_1|/2 = |V(w^1_1)|/2\geq n/(4d+4)$. Assign 0 to all variables except those in $U_1\setminus U_2$, we have a circuit of at least $n/(4d+4)$ variables, whose degree is at most $d-1$ because $g_1$ is satisfied and removed. If $|U_2| \geq |U_1|/2$, then proceed similarly: if $|U_3| \leq |U_2|/2$, then assign 0 to all variables except those in $U_2 \setminus U_3$, we have a circuit of at least $|U_2|/2$ variables, whose degree is at most $d-1$ because $g_1$ is satisfied and removed. Proceed in this fashion, so we either have a circuit of at least $\frac{n/(d+1)}{2^s} \geq \frac{n/(d+1)}{2^{2d(d+1)}}$ variables whose degree is at most $d-1$, or we end up with $|U_s| > \frac{n/(d+1)}{2^s} \geq \frac{n/(d+1)}{2^{2d(d+1)}} >0$, which is impossible because any variable in $U_s$ is a zero-variable of $C$. \\
\end{itemize}
\end{itemize}

This completes the description of the procedure.

Note that no zero-variables are created in any of the above cases. This is true because in all cases except {\bf Case 1}, we assign the variables in $C$ only the value 0, which does not create zero-variables, while in {\bf Case 1}, $x_1$ is assigned 1, but every gate containing $x_1$ is removed (except the output gate). Note also that the second base case of $t=3$ can be treated by the above process because $t=3$ is only possible in {\bf Case/Subcase} 1, 2.2, and 2.3.2, all of which do not rely on the inductive hypothesis.

So in one iteration of the above process, we either: (1) reduce the number of variables from $n$ to  $n/(d+1)$ and assign 1 to a variable ({\bf Case-1} operation), or (2) reduce the number of  variables from $n$ to a number of variables that is at least $\min\{\log^{(d^{t-2})}(\frac{n}{2d+2}), n/(2d+2), \log^{(d^{t-2})}(\frac{n}{4d(d+1)^2}), \frac{n/(d+1)}{2^{2d(d+1)}}\} = \log^{(d^{t-2})}(\frac{n}{4d(d+1)^2})$, and reduce the degree of the variables by 1 ({\bf Case-2} operation). Afterwards, we can repeat the process until $f(n)$ variables are assigned 1, or until the degree of the variables in the circuit is at most 1. After a number of iterations, if $f(n)$ variables are already assigned 1 and the circuit is not empty, then we can assign 0 to all other variables and we are done. If the degree of the variables in the circuit is at most 1, then as we mentioned at the beginning of the proof, at least half of the remaining variables can be assigned 1. So it remains to be shown that when the degree of the variables in the circuit is at most 1, there are at least $2f(n)$ variables left.

In any sequence of iterations, {\bf Case-1} operation is applied at most $f(n)$ times and {\bf Case-2} operation is applied at most $d$ times. Let $g(n) = n/(d+1)$ and $h(n) = \log^{(d^{t-2})}\left(\frac{n}{4d(d+1)^2}\right)$. Note that $g(h(n)) \leq h(g(n))$, i.e., $g\circ h \leq h \circ g$. So the number of variables in the circuit after any sequence of iterations is at least: $$\underbrace{g \circ \ldots \circ g}_{f(n)} \circ \underbrace{h \circ \ldots \circ h}_{d}(n).$$

Note that $h(n) = \log^{(d^{t-2})}\left(\frac{n}{4d(d+1)^2}\right) \geq \log^{(d^{t-2})}\log n = \log^{(d^{t-2}+1)} n$. So we have:
\begin{eqnarray}
\underbrace{h \circ \ldots \circ h}_{d}(n) \geq \underbrace{\log^{(d^{t-2}+1)} \circ \ldots \circ \log^{(d^{t-2}+1)}}_{d} n = \log^{(d(d^{t-2}+1))} n = \log^{(d^{t-1}+d)} n.\label{hs}
\end{eqnarray}
On the other hand:
\begin{eqnarray}
\underbrace{g \circ \ldots \circ g}_{f(n)}(n) = n/(d+1)^{f(n)} = n/(d+1)^{\log^{(d^{t})}n} \geq n/\log^{(d^{t}-2)}n > \log n.\label{gs}
\end{eqnarray}

Finally, since $d\geq 2$ and $t\geq 3$, we have:
\begin{eqnarray}\underbrace{g \circ \ldots \circ g}_{f(n)} \circ \underbrace{h \circ \ldots \circ h}_{d}(n) \geq \log(\log^{(d^{t-1}+d)} n) = \log^{(d^{t-1}+d+1)} n \geq 2\log^{(d^t)} n = 2f(n).\end{eqnarray}

This means that in any sequence of iterations, we either assign 1 to $f(n)$ variables or end up with at least $2f(n)$ variables of degree at most 1, in which case the circuit can be satisfied by assigning 1 to $f(n)$ variables. So in either case, the statement is true for circuits of depth $t\geq 2$.

This completes the proof.
\end{proof}

\begin{theorem}
\label{thm:main1}
The {\sc wsat$^-[t]$} ($t \geq 2)$ problem on circuits of genus $g(n) =n^{o(1)}$ ($n$ is the number of variables in the circuit) is $FPT$, and is $W[t]$-complete for odd $t$ and $W[t-1]$-complete for even $t$ if $g(n) = n^{O(1)}$.
\end{theorem}

\begin{proof}
Let $g(n) = n^{o(1)}= n^{1/\mu(n)}$, where $\mu(n)$ is a complexity function, and let $(C, k)$ be an instance of the {\sc wsat$^-[t]$} ($t \geq 2)$ problem on circuits of genus $g(n)$. By Proposition~\ref{prop:removezeros}, we can assume that $C$ has no zero-variables, and that the number of variables $n$ in $C$ is least $g(n)$. By Lemma~\ref{lem:occurrences}, we may assume that the number of occurrences of the literals in $C$ is $O(n)$; if this is not the case then the genus of the circuit is not upper bounded by $g(n)$, and we reject the instance. By Theorem~\ref{thm:sparse}, $C$ has a satisfying assignment in which at least $f(n)$ variables are assigned the value 1, where $f(n)$ is the function given in the lemma. Therefore, if $k \leq f(n)$ then we accept the instance $(C, k)$; otherwise, $k > f(n)$ and in fpt-time we can decide the instance by a brute-force algorithm that enumerates every weight-$k$ assignment.

The hardness results for $g(n) = n^{O(1)}$ follow from Theorem~\ref{lem:hardnessantimonotone}.
\end{proof}

\section{The monotone case}
\label{sec:monotone}

In this section, we give a complete characterization of the parameterized complexity of the {\sc wsat$^+$} problem (i.e., {\sc wsat$^+[2]$}) with respect to the genus of the circuit, and a partial characterization of the parameterized complexity of
{\sc wsat$^+[t]$} ($t \geq 2)$. We start with the following hardness result:

\begin{theorem}\label{lem:hardnessmonotone}
Let $c > 0$ be a constant. The {\sc wsat$^+[t]$} ($t \geq 2)$ problem on circuits of genus $g(n) =\Omega(n^c)$, where $n$ is the number of variables in the circuit, is $W[t]$-complete for even $t$ and $W[t-1]$-complete for odd $t$.
\end{theorem}

\begin{proof}
To prove the hardness result in the theorem, we show that {\sc wsat$^+[t]$} is fpt-reducible to {\sc wsat$^+[t]$} on circuits of genus $g(n) =\Omega(n^c)$. Since {\sc wsat$^+[t]$} is $W[t]$-hard for even $t$, and $W[t-1]$-hard for odd $t > 1$~\cite{fptbook}, the hardness result follows. Suppose that $g(n) = c'n^c$, for some constant $c' > 0$.

Let $(C_0, k)$ be an instance of {\sc wsat$^+[t]$}. Suppose that $C_0$ has $n_0$ variables and $m_0$ gates (including the variables). Therefore, the genus of $C_0$ is at most $m_0^2$. If $m_0^2 \leq c'n_0^c$, then the fpt-reduction outputs the instance $(C, k)$, where $C= C_0$. If $m_0^2 > c'n_0^c$, let $C$ be the circuit obtained from $C_0$ by adding $\lceil(m_0^2/c')^{(1/c)}\rceil - n_0$ variables incoming to an {\sc or}-gate that is incoming to the output {\sc and}-gate of $C_0$. The fpt-reduction outputs the instance $(C, k+1)$. Obviously, the genus of $C$ is at most that of $C_0$, which is at most $m_0^2$. It can be easily verified that the genus of $C$, in both cases, is at most $c'n^c$, where $n$ is the number of variables in $C$. It is easy to see that $C_0$ has a weight-$k$ satisfying assignment if and only if $C$ has a weight-$(k+1)$ satisfying assignment. It follows that the above reduction is an fpt-reduction from {\sc wsat$^+[t]$} to {\sc wsat$^+[t]$} on circuits of genus $g(n) =\Omega(n^c)$.

The completeness of the problem follows from the membership of {\sc wsat$^+[t]$} in $W[t]$ for even $t$, and the membership of {\sc wsat$^+[t]$} in $W[t-1]$ for odd $t > 1$.
\end{proof}

\begin{theorem}
\label{thm:general}
Let $(C, k)$ be an instance of {\sc wsat$^+[t]$} ($t \geq 2$) such that $C$ has genus $g(n) = n^{o(1)}$, where $n$ is the number of variables in $C$. There is an fpt-time algorithm that either decides the instance $(C, k)$ correctly, or reduces it to $h(k)n^{O(1)}$ many instances $(C', k')$ of {\sc wsat$^+[t]$}, where $h$ is a complexity function of $k$ and $k' \leq k$, such that $(C, k)$ is a yes-instance if and only if at least one of the instances $(C', k')$ is, and such that each instance $(C', k')$ satisfies that: (1) the number of critical gates in $C'$ is at most $2k'$, (2) every variable in $C'$ is incoming to gates in at most two subtrees $T_p, T_q$ of $T_C'$ rooted at critical gates $p, q$ in $C'$, and (3) the genus of $C'$ is at most $g(n)$.
\end{theorem}

\begin{proof}
Let $g(n)$ be a complexity function such that $g(n)=n^{o(1)}$. Since $g(n) = n^{o(1)}$, $g(n)\leq n^{1/\mu(n)}$ for some complexity function $\mu(n)$.

Let $(C, k)$ be an instance of {\sc wsat$^+[t]$}, where $C$ is a $\Pi_t^+$ circuit with set of variables $X=\{x_1, \ldots, x_n\}$, and $k$ is the parameter. If more than $k$ variables are incoming to the output gate of $C$, then clearly $C$ has no satisfying assignment of weight $k$, and we reject the instance $(C, k)$. Otherwise, we can assign the value 1 to the variables incoming to the output-gate of $C$, remove these variables, and update $C$ and $k$ accordingly. So we may assume, without loss of generality, that $C$ has no variables incoming to its output gates, and that all gates incoming to the output gates are {\sc or}-gates (by the simplification rules discussed in Section~\ref{sec:prelim}), and hence are critical gates.

For each critical gate $p$ in $C$, consider the subtree $T_p$ of $T_C$. In the case when $t=2$, this subtree is trivial, and consists of gate $p$. We form an auxiliary graph ${\cal B}$ as follows. Starting at each critical gate $p$, we contract the edges in $T_p$ to form a single vertex $p'$ whose incoming variables are the variables that are incoming to at least one gate in $T_p$. Note that if a variable is incoming to several gates in $T_p$, then there will be multiple edges between $p'$ and this variable. Let ${\cal G}$ be the set of vertices resulting from contracting each tree $T_p$ corresponding to a critical gate $p$ in $C$. Let ${\cal B} = ({\cal G}, X)$ be the underlying undirected bipartite graph resulting from this contraction with the multiple edges removed. That is, there is an (undirected) edge in ${\cal B}$ between a variable $x_i \in X$ and a gate $p'$ in ${\cal G}$ if and only if $x_i$ is incoming to some gate in $T_p$. Clearly, the genus of ${\cal B}$ is at most $g(n)$. Observe that since each critical gate $p$ must be satisfied by every assignment that satisfies $C$, for any vertex $p'$ in ${\cal G}$, at least one variable incident to $p'$ in ${\cal B}$ must be assigned 1 in any truth assignment satisfying $C$. Pose $n_g = |{\cal G}|$.

We partition the variables in $X$ into two sets: $X_{\geq 3}$ that consists of each variable in $X$ whose degree in $B$ is at least 3, and $X_{\leq 2}$ consisting of each variable in $X$ whose degree in $B$ is at most 2. Pose $n_3 = |X_{\geq 3}|$ and $n_2 = |X_{\leq 2}|$. By defining a multihypergraph whose vertex-set is ${\cal G}$, and whose hyperedges correspond to the neighborhoods of the variables in $X_{\geq 3}$, we obtain from Lemma~\ref{lem:edgebound} that $n_3 \leq 2n_g + 4g(n)$; if the preceding upper bound on $n_3$ does not hold, then we reject the instance (this means that the genus of the circuit is not at most $g(n)$). We perform the following search-tree algorithm ${\cal A}$ that distinguishes two cases:  \\

\noindent{\bf Case 1.} $n_g \leq n^{1/\mu(n)}$. In this case we have $n_3 \leq  2n_g + 4g(n) \leq 6 n^{1/\mu(n)}$. The number of subsets of $X_{\geq 3}$ of size at most $k$ is at most $\Sigma_{i=0}^{k} {n_3 \choose i} \leq kn_3^{k} \leq k\cdot(6n^{1/\mu(n)})^k$. We try each such subset of $X_{\geq 3}$ as a candidate subset of variables that will be assigned value 1 by a satisfying assignment of weight $k$. For each such candidate subset $S$, we update the gates in $C$ in a natural way according to the partial assignment assigning the variables in $S$ the value 1, and those in $X_{\geq 3}\setminus S$ the value 0. We remove all variables in $X_{\geq 3}$ from $C$, and update $C$ and $k$ appropriately. Since each remaining variable is in $X_{\leq 2}$, each variable can satisfy at most 2 critical gates, and hence if the number of critical gates in $C$ is more than $2k$, then we can reject the resulting instance $(C, k)$. Therefore, for each instance resulting from the enumeration of such a subset $S$ of $X_{\geq 3}$, either the number of remaining critical gates in $C$ is more than $2k$ and we reject the instance since $k$ variables in $X_{\leq 2}$ cannot satisfy all the critical gates of $C$, or the number of critical gates in $C$ is at most $2k$. Since the number of enumerated candidate subsets of $X_{\geq 3}$ is at most $k\cdot (6n^{1/\mu(n)})^k$, the statement of the theorem follows from Lemma~\ref{lem:fpttime}.

\noindent {\bf Case 2.} $n_g > n^{1/\mu(n)}$. Let $G$ be the subgraph of ${\cal B}$ induced by the set of vertices in ${\cal G}$ plus those in $X_{\geq 3}$. Since $n_3 \leq 2n_g + 4g(n) \leq 6 n_g$, the number of vertices in $G$ is at most $7n_g$. Since the genus of $G$ is at most $g(n)$, by Lemma~\ref{lem:Euler}, the number of edges in $G$ is at most $21n_g + 6g(n) \leq 27 n_g$. Let $Y_{\geq 3}$ be the set of variables in $X_{\geq 3}$ of degree at least $27n_g/\log{n}$ in $G$. Since the number of edges in $G$ is at most $27n_G$, it follows that $|Y_{\geq 3}| \leq \log{n}$. In time $(\log{n})^k$, which is fpt-time by Lemma~\ref{lem:fpttime}, we can enumerate each subset of $Y_{\geq 3}$ of size at most $k$ as a candidate subset of variables that are assigned value 1 by a satisfying assignment of weight $k$. For each such {\em nonempty} candidate subset, $C$ is updated appropriately (as in {\bf Case 1} above) and $k$ is decreased by at least the size of the subset, which is nonzero, and we can repeat the execution of the whole algorithm ${\cal A}$; this algorithm will be repeated at most $k$ times. If the candidate subset is empty, then along this branch we reject the instance $(C, k)$ since $C$ cannot be satisfied by an assignment of weight $k$. The preceding statement can be justified as follows. In any satisfying assignment, the critical gates, whose number is $n_g > n^{1/\mu(n)}$, must be satisfied. Since the chosen subset of $Y_{\geq 3}$ is empty, we are working under the assumption that no variable in $Y_{\geq 3}$ is assigned 1 by any satisfying assignment. Therefore, the variables assigned 1 by any satisfying assignment must be chosen from $X_{\geq 3} - Y_{\geq 3}$ or from $X_{\leq 2}$. Each variable in $X_{\geq 3} - Y_{\geq 3}$ can satisfy at most $27n_g/\log{n}$ critical gates in $C$, and each variable in $X_{\leq 2}$ can satisfy at most 2 critical gates. Therefore, $k$ variables from $(X_{\geq 3} - Y_{\geq 3}) \cup X_{\leq 2}$ can satisfy at most $27kn_g/\log{n} < n_g$ critical gates in $C$, and hence cannot satisfy $C$. We assumed here that $k < \log{n}/27$; otherwise, we can decide the instance in fpt-time from the beginning.

It follows that the algorithm ${\cal A}$ outlined above runs in fpt-time, and either solves the instance $(C, k)$, or reduces it to $h(k)n^{O(1)}$ many instances $(C', k')$ ($k' < k$), such that $(C, k)$ is a yes-instance if and only if at least one of the instances $(C', k')$ is, and such that each of the instances $(C', k')$ satisfies conditions (1), (2), and (3) in the statement of the theorem.
\end{proof}

\begin{theorem}
\label{cor:wsat}
The {\sc wsat$^+$} problem on circuits of genus $g(n) = n^{o(1)}$ is $FPT$.
\end{theorem}

\begin{proof}
By Theorem~\ref{thm:general}, in fpt-time we can reduce an instance $(C, k)$ of {\sc wsat$^+$} on circuits of genus $g(n) = n^{o(1)}$ to $h(k)n^{O(1)}$ many instances $(C', k')$ ($k' < k$) of {\sc wsat$^+$}, where $h$ is a complexity function of $k$ and $k' \leq k$, such that $(C, k)$ is a yes-instance if and only if at least one of the instances $(C', k')$ is, and such that each instance $(C', k')$ satisfies that: (1) the number of critical gates in $C'$ is at most $2k'$, and (2) every variable in $C'$ is incoming to gates in at most two subtrees $T_p, T_q$ of $T_C'$, rooted at critical gates $p, q$ in $C'$. Therefore, it suffices to show that we can decide each such instance $(C', k')$ in fpt-time.

First, observe that since each subtree $T_p$ rooted at a critical gate $p$ consists of a single critical gate of $C'$, each variable in $C'$ has outdegree at most 2; that is, each variable in $C'$ is incoming to at most two gates in $C'$.

For two variables $x_i$ and $x_j$ in $C'$, if the set of gates that $x_i$ is incoming to is a subset of that of $x_j$, then we say that $x_j$ {\em dominates} $x_i$.
We perform the following reductions. If more than $k'$ variables are incoming to the output gate of $C'$, then clearly $C'$ has no satisfying assignment of weight at most $k'$, and we reject the instance $(C', k')$. Otherwise, we can assign the value 1 to the variables incoming to the output gate of $C'$, remove them, and update $C'$ and $k'$ accordingly. For any two 2-literal gates that have the same pair of variables incoming to them, we remove one of the two gates from $C'$. So assume, without loss of generality, that the given instance $(C', k')$ satisfies that $C'$ contains no variables that are incoming to its output gate, and that there are no two 2-literal gates in $C'$ that have the same pair of variables incoming to them. For every two variables $x_i$ and $x_j$ in $C$, if $x_i$ dominates $x_j$ then remove $x_j$. After applying the previous reductions, it is easy to see that the number of degree-1 variables is at most $2k'$, and the number of degree-2 variables is at most ${2k' \choose 2}$. It follows that the resulting circuit has size $O(k'^2)$, and in fpt-time we can decide if $C'$ has a satisfying assignment of weight $k'$. This completes the proof.
\end{proof}

Combining the above theorem with Theorem~\ref{lem:hardnessmonotone}, we obtain a complete characterization of the parameterized complexity of {\sc wsat$^+$} in terms of the genus of the circuit:

\begin{theorem}
\label{thm:mainmonotonewsat}
The {\sc wsat$^+$} problem on circuits of genus $g(n)$ is $FPT$ if $g(n) = n^{o(1)}$, and $W[2]$-complete if $g(n) = n^{O(1)}$, where $n$ is the number of variables in the circuit.
\end{theorem}

We follow the exact terminology of~\cite{dht}. Let $G$ be a graph, and let $V' \subseteq V(G)$ and $E' \subseteq E(G)$ be such that every vertex in $V'$ is an endpoint of some edge in $E'$. Let $G^-$ be the graph obtained
from $G$ by removing the vertices in $V'$ and the edges in $E'$. $G$ is said to be $(V', E')$-{\em embeddable} (in the plane) if $G^-$ is embeddable in the plane. The vertices in $V'$ and the edges in $E'$ are called {\em flying}. The flying edges are partitioned into: (1) {\em bridges}, those are the edges whose both endpoints are in $G^-$; (2) {\em pillars}, those are the edges with exactly one endpoint in $G^-$; and (3) {\em clouds}, those are the edges whose both endpoints are not in $G^-$ (i.e., are in $V'$).
A {\em partially triangulated $(r \times r)$-grid} is a graph that contains the $(r \times r)$-grid as a subgraph, and is itself a subgraph of a triangulation of the $(r \times r)$-grid. A graph $G$ is called an
$(r, \ell)$-{\em gridoid} if it is $(V', E')$-embeddable for some $V', E'$ such that $G^-$ is a partially triangulated $(r'\times r')$-grid for some $r' \geq r$, and $E'$ contains at most $\ell$ edges and no clouds (i.e., $E'$ contains no edges whose both endpoints are in $V'$). The following result was proved in~\cite{dht}:

\begin{theorem}[\cite{dht}]
\label{thm:dht}
If a graph $G$ of genus $g$ excludes all $(\lambda - 12g, g)$-gridoids as contractions, for some $\lambda \geq 12g$, then the branchwidth of $G$ is at most $4\lambda(g+1)$.
\end{theorem}

We shall assume, without loss of generality, that in any instance $(C, k)$ of the problem, $k < g(n)$ and $k$ is larger than any prespecified constant; those instance that violate either of the preceding conditions can be decided in fpt-time (for some fixed time complexity function that depends on $g(n)$ and the size of $C$). We have the following result:

\begin{lemma}\label{lem:bw}
Let $(C, k)$ be an instance of {\sc wsat$^+[t]$} ($t \geq 2)$ such that $C$ has genus $g(n)$ and at most $2k$ critical gates, where $n$ is the number of variables in $C$. Let $C^-$ be the circuit resulting from $C$ after removing the output gate. The branchwidth of the underlying graph of $C^-$ is $O(g^2(n))$.
\end{lemma}

\begin{proof}
We will show that the underlying graph of $C^-$ excludes all $(\lceil \sqrt{kg(n)}\rceil, g(n))$-gridoids as contractions. By setting $\lambda = 12g(n) + \lceil \sqrt{kg(n)}\rceil$, the result follows from Theorem~\ref{thm:dht}. (We assumed that $k  < g(n)$.)

Suppose, to get a contradiction, that the underlying graph of $C^-$ contains an $(r, g(n))$-gridoid $G$ as a contraction, for some integer $r \geq \lceil \sqrt{kg(n)} \rceil$. Since the depth of $C$ is at most $t$, every literal and gate in $C^-$ is within distance (i.e., length of a shortest path) at most $t$ from some critical gate of $C$. Let $S$ be the set of vertices in $G$, each of which either corresponds to a critical gate of $C$ or to a contraction of a critical gate of $C$, and note that $|S| \leq 2k$. Clearly, every vertex in $G$ must be within distance at most $t$ from one of the vertices in $S$. Embed $G$ in the plane, and let $G^-$, $E'$ and $V'$ be as in the definition of the gridoid. Note that $E'$ contains at most $g(n)$ edges, and each edge of $E'$ must be incident to at least one vertex in $G^-$. Call an endpoint of an edge in $E'$ that is in $G^-$ an {\em anchor vertex}. Since $|E'| \leq g(n)$, it follows that the number of anchor vertices is at most $2g(n)$. There is a path of length at most $t$ from every vertex $v$ in the partially-triangulated grid $G^-$ to some vertex in $S$; fix such a path for every vertex $v$ in $G^-$, and denote it by $P_v$. Since the number of grid vertices at distant at most $t$ from some grid vertex is $O(t^2)$, the number of paths $P_v$ that pass through a fixed anchor vertex is $O(t^2)$. Therefore, the number of grid vertices $v$ whose paths $P_v$ go through anchor vertices is $O(t^2) \cdot 2g(n) =O(g(n))$. For any other vertex $v$, its path $P_v$ lies completely within $G^-$, and hence the number of such vertices $v$ is $O(t^2) \cdot |S| =O(k)$. Since for every vertex $v$ in $G^-$, $P_v$ either goes through an anchor vertex or lies completely within the grid, the number of grid vertices is at most $O(g(n)) + O(k)= O(g(n))$ (we assumed that $k  < g(n)$). Since the number of vertices in $G^-$ is at least $r^2 = \Omega(kg)$, this is a contradiction since $k$ can be chosen to be larger than any prespecified constant, and in such case there would be grid vertices that are not within distance $t$ from any vertex in $S$.
\end{proof}

\begin{definition}
Let $G=(V, E)$ be a graph. A {\it tree decomposition} of $G$ is
a pair $({\cal V}, {\cal T})$ where ${\cal V}$ is a collection of
subsets of $V$ such that $\bigcup_{X_i \in {\cal V}} = V$,
and ${\cal T}$ is a tree whose node set is $\cal V$, such that:

\begin{enumerate}
\item for every edge $\{u, v\} \in E$, there is an $X_i \in {\cal V}$, such
 that $\{u, v\} \subseteq  X_i$;

\item for all $X_i, X_j, X_k \in {\cal V}$, if the node $X_j$ lies on the
   path between the nodes $X_i$ and $X_k$ in the tree ${\cal T}$, then
   $X_i \cap X_k \subseteq X_j$;

The {\it width} of the tree decomposition $({\cal V}, {\cal T})$ is
defined to be $\max\{|X_i| \mid X_i \in {\cal V} \} - 1$. The
{\it treewidth} of the graph $G$ is the minimum tree width over all
tree decompositions of $G$.
\end{enumerate}
\end{definition}

A tree decomposition $({\cal V}, {\cal T})$ is {\it nice} if it satisfies
the following conditions:

\begin{enumerate}

\item Each node in the tree ${\cal T}$ has at most two children.

\item If a node $X_i$ has two children $X_j$ and $X_k$ in the tree ${\cal T}$,
   then $X_i=X_j=X_k$; in this case node $X_i$ is called a {\em join node}.

\item If a node $X_i$ has only one child $X_j$ in the tree ${\cal T}$, then either
   $|X_i|=|X_j|+1$ and $X_j \subset X_i$, and in this case $X_i$ is called an {\em insert node}; or $|X_i|=|X_j|-1$ and
   $X_i\subset X_j$, and in this case $X_i$ is called a {\em forget node}.
\end{enumerate}

\begin{theorem}
\label{thm:treewidth}
Let $C$ be a $\Pi_t^+$ circuit, and let $G = (V,E)$ be the undirected underlying graph of $C$ with the output gate removed. If a tree decomposition of width $\omega$ for $G$ is given, then a minimum weight satisfying assignment of $C$ can be computed in time $2^{O(\omega)}N^{O(1)}$, where $N$ is the number of nodes in the given tree decomposition.
\end{theorem}

\begin{proof}
Let $\mathcal{X} = \langle \{X_i \mid i \in {\cal T} \}, {\cal T} \rangle$ be a {\em nice} tree decomposition for the graph $G$. We assume that the tree decomposition is nice; otherwise, based on ${\cal T}$ we can compute a nice tree decomposition of the same width in polynomial time in the size of ${\cal T}$~\cite{kloks}. To simplify the notation, we call a vertex in $G$ a ``variable" (resp. a ``gate") if its corresponding vertex in $C$ is a variable (resp. a gate).

We use a dynamic programming approach to compute a minimum weight satisfying assignment for $C$. Let $X_i = (x_{i_1}, \ldots, x_{i_{n_i}})$ be a bag in $\mathcal{X}$, where each of $x_{i_1}, \ldots, x_{i_{n_i}}$ is either a variable or a gate. For an $x_{i_r} \in X_i$, if $x_{i_r}$ is a variable we assign it either the color ``white'', meaning that its value is 0/{\sc false}, or the color ``black'', meaning that its value is 1/{\sc true}; if $x_{i_r}$ is a gate, we assign it one of three colors: ``black'', ``gray'', and ``white''. Here are the interpretations of the colors, and the rules for assigning them to the gates:

\begin{itemize}
\item black: {\sc true} and justified. For an {\sc or}-gate, this means that one of the vertices incoming to $x_{i_r}$ is colored black or gray; for an {\sc and}-gate, this means that all the vertices incoming to $x_{i_r}$ are black or gray.

\item gray: {\sc true} but unjustified.  For an {\sc or}-gate, this means that every vertex incoming to $x_{i_r}$ is either uncolored or is colored white; for an {\sc and}-gate, this means that at least one vertex incoming to $x_{i_r}$ is uncolored, and the colored vertices incoming to $x_{i_r}$ are black or gray.

\item white: {\sc false}, either justified or unjustified.  For an {\sc or}-gate, this means that every vertex incoming to $x_{i_r}$ is either uncolored or is colored white; for an {\sc and}-gate, this means that one of the vertices incoming to $x_{i_r}$ is either uncolored or is white.

\end{itemize}

A vector $c_i = (c_{i_1}, \ldots, c_{i_{n_i}})$ is called a {\em coloring} of the bag $X_i$, where $c_{i_r}$ is the color of $x_{i_r}$. The weight of a coloring $c_i$ of a bag $X_i$, denoted $W(c_i)$, is the minimum number of variables assigned {\sc true} in the graph induced by the subtree of ${\cal T}$ rooted at $X_i$, under the restriction that $c_i$ is the coloring of $X_i$.

The dynamic programming algorithm will compute valid colorings of the bags in ${\cal T}$ and their weights in a bottom-up fashion starting at the leaves of ${\cal T}$. During this process, we check for validity of the colorings according to the rules of assigning the colors, and purge invalid ones. Additionally, if a gate in $G$ is critical and is colored white, then the coloring is also invalid and purged.

First, for each leaf bag in the tree decomposition, we compute the valid colorings and their weights for this bag. The valid colorings can be computed by enumerating all colorings and checking for their validity according to coloring rules; this takes time $2^{O(\omega)}N^{O(1)}$. Next, we move up the tree from the leaves to the root, computing the colorings and their weights of a parent depending on the colorings and weights of its child (or children). We set the following ground rule regarding the coloring of a vertex shared by a parent (bag) and its child (bag):

\begin{quote}
{\bf Ground Rule}: If the shared vertex is a variable, then its color must be the same in the parent and in the child; if the shared vertex is a gate, then either its color is the same in the parent and in the child, or its color is gray ({\sc true} but unjustified) in the child and black ({\sc true} and justified) in the parent.
\end{quote}
The ground rule is based on the following reasoning: A vertex that is colored black or white does not change its color in a valid coloring; an {\sc and}-gate colored gray can be {\em upgraded} (later) to black when all vertices incoming to it are colored black or gray; and an {\sc or}-gate colored gray can be {\em upgraded} to black when a vertex incoming to it becomes black or gray.

We distinguish three cases according to the types of the nodes in the tree decomposition.

\begin{enumerate}
\item Forget node: Let $X_i$ be the bag of a forget node and $X_j = X_i \cup \{x\}$ be the bag of its child, where $x$ is the vertex to be ``forgotten". The colorings of $X_i$ are the projection of the colorings of $X_j$. The weight of a coloring $c$ of $X_i$ is the minimum weight of the colorings of $X_j$ that produce $c$. Note that by the time a gate $g$ is to be forgotten, it will not be colored gray because by then all vertices incoming to $g$ have been considered, and hence its color should not remain unjustified.

\item Insert node: Let $X_i$ be the bag of an insert node and $X_j = X_i \setminus \{x\}$ be the bag of its child, where $x$ is the vertex to be ``inserted". We will extend the colorings of $X_j$ by assigning $x$ its possible color options. After inserting the new vertex and assigning it a color, a coloring may become invalid, and in which case the coloring is discarded. After inserting the new vertex and assigning it a color, it is possible that some gray gate may be upgraded to black, then it is updated as such. Note that upgrading the color of a vertex $v$ from gray to black does not affect the colors of the vertex that $v$ incoming to (by the coloring rules). The weight of a coloring $c$ of $X_i$ is the minimum weight of the colorings of $X_j$ that produce $c$, plus one if the new vertex $x$ is a variable and is assigned {\sc true}.

\item Join node: Let $X_i$ be the bag of a join node and $X_j$, $X_k$ be the bags of its children, where $X_i = X_j = X_k$. Let $x$ be a vertex in $X_i$. If $x$ is a variable, then the color of $x$ must be the same in $X_i$, $X_j$, and $X_k$ according to the {\bf Ground Rule}. If $x$ is a gate, the color of $x$ can be the same in $X_i$, $X_j$, and $X_k$, or, according to the {\bf Ground Rule}, one of the following cases applies: (1) $x$ is black in $X_i$ and $X_j$, and gray in $X_k$ (or symmetrically, black in $X_i$ and $X_k$, and gray in $X_j$), or (2) $x$ is black in $X_i$, and gray in both $X_j$ and $X_k$. In the following, we discuss these cases based on the type of the gate $x$.

    If $x$ is an {\sc and}-gate, case (1) happens when all the vertices incoming to $x$ are {\sc true} (either justified or unjustified) and all of them are in the subtree rooted at $X_j$ (and hence $X_i$), but not all of them are in the subtree rooted at $X_k$. Case (2) happens when all the vertices incoming to $x$ are {\sc true} (either justified or unjustified), and all of them are in the subtree rooted at $X_i$, but not all of them are in the subtree rooted at $X_j$ or $X_k$.

    If $x$ is an {\sc or}-gate, case (1) happens when a vertex incoming to $x$ is {\sc true} in the resolved portion of the subtree rooted at $X_j$ (and hence $X_i$), but it is not in the subtree rooted at $X_k$. Case (2) is impossible because if $x$ has a vertex incoming to it that is colored black or gray, this vertex should appear in $X_j$ or in $X_k$.

    In each of these cases, if a coloring $c_i$ of $X_i$ is produced from a coloring $c_j$ of $X_j$ and a coloring $c_k$ of $X_k$, then $W(c_i) = \min(W(c_j)+W(c_k) - \#_1(c_i))$ over all colorings $c_j$ and $c_k$ that produce $c_i$, where $\#_1(c_i)$ is the number of variables assigned {\sc true} in coloring $c_i$.

\end{enumerate}
In each of the these three cases, the running time is $2^{O(\omega)}N^{O(1)}$.

Finally, the minimum weight satisfying assignment is the minimum weight of the colorings of the root. The total running time of the dynamic programming algorithm outlined above is $2^{O(\omega)}N^{O(1)}$.
\end{proof}

\begin{theorem}
\label{thm:mainmonotone}
The {\sc wsat$^+[t]$} problem $(t > 2)$ on circuits of genus $g(n) = O(\sqrt{\log{n}})$ is $FPT$, where $n$ is the number of variables in the circuit.
\end{theorem}

\begin{proof}
Let $(C, k)$ be an instance of {\sc wsat$^+[t]$} on circuits of genus $g(n) \leq c\sqrt{\log{n}}$, for some fixed (known) constant $c > 0$. By Theorem~\ref{thm:general}, in fpt-time we can reduce the instance $(C, k)$ to $h(k)n^{O(1)}$ many instances $(C', k')$ of {\sc wsat$^+[t]$}, where $h$ is a complexity function of $k$ and $k' \leq k$, such that $(C, k)$ is a yes-instance if and only if at least one of the instances $(C', k')$ is, and such that $C'$ has at most $2k'$ critical gates. Therefore, without loss of generality, we may assume that $C$ has at most $2k$ critical gates. By Lemma~\ref{lem:bw}, the branchwidth of $C$ is at most $c_1\log{n}$, for some fixed constant $c_1 > 0$, and hence, by the results of Robertson and Seymour~\cite{rs}, the treewidth of $C$ is at most $c_2 \log{n}$ for some fixed constant $c_2 > 0$. Using the algorithm of Amir~\cite{amir}, we can decide if the treewidth of $C$ is at most $c_3 \log{n}$ for some fixed constant $c_3 > 0$ (if not, the genus does not satisfy the assumed upper bound and we reject the instance), and if so, the algorithm in~\cite{amir} returns a tree decomposition of $C$ of width $c_4\log{n}$, for some constant $c_4 > 0$, in time $2^{O(\log{n})}|C|^{O(1)}=|C|^{O(1)}$. By applying Theorem~\ref{thm:treewidth}, we conclude that we can decide the instance $(C, k)$ in fpt-time.
\end{proof}

\section{Concluding remarks}
\label{sec:conclusion}
In this paper we tried to characterize the parameterized complexity of the canonical monotone and antimonotone {\sc wsat$[t]$} problems in terms of the genus of the circuit. For {\sc wsat$^-[t]$}, the characterization we provided is precise.
For {\sc wsat$^+[t]$}, however, there is still a big gap between the two genus bounds of $O(\sqrt{\log{n}})$ and $n^{o(1)}$. Closing this gap, or even reducing it, is a very interesting question that we leave open.
We mention that several graph problems, including {\sc independent set} on graphs/hypergraphs, {\sc hitting set}, and {\sc red/blue dominating set}, can be reduced to the {\sc wsat$^-[t]$} and {\sc wsat$^+[t]$} problems via fpt-reductions that preserve the genus of the underlying graph. Therefore, the fixed-parameter tractability results for
{\sc wsat$^-[t]$} and {\sc wsat$^+[t]$} obtained in the current paper imply fixed-parameter tractability results for those problems on graphs whose genus satisfies the upper bound requirements.

Similar characterizations of the subexponential-time computability of {\sc wsat$^-[t]$} and {\sc wsat$^+[t]$} in terms of the genus can be obtained. It is not difficult to prove by combining some results in this paper with a standard divide-and-conquer approach based on the separator theorem in~\cite{djidjev}, that {\sc wsat$^-[t]$} and {\sc wsat$^+[t]$} are solvable in subexponential-time if the genus is $o(n)$, and that they are not solvable in subexponential-time if the genus is $\Omega(n)$ unless the exponential-time hypothesis ({\sc ETH})
fails.  We refer the reader to~\cite{genus} for examples of how this standard approach can be applied to obtain such subexponential-time computability results. It would be interesting to see if any characterization of the approximation of the optimization versions of {\sc wsat$^-[t]$} and {\sc wsat$^+[t]$} based on the genus of the circuit can be derived. We also leave this is an open question.

\bibliographystyle{plain}
\bibliography{ref}

\end{document}